\documentclass[conference]{IEEEtran}

\usepackage{cite}
\usepackage{amsmath,amssymb,amsthm}
\usepackage{pgfplots}
\pgfplotsset{compat=1.18}

\makeatletter
\let\MYcaption\@makecaption
\makeatother
\usepackage[font=footnotesize]{subcaption}
\makeatletter
\let\@makecaption\MYcaption
\makeatother

\pgfplotsset{
    every axis/.append style={
        label style={font=\footnotesize},
        tick label style={font=\footnotesize},
        legend style={font=\footnotesize}
    }
}

\usepackage[pdftex,bookmarks,colorlinks,breaklinks, linkcolor=blue, citecolor=blue,
            urlcolor=blue]{hyperref}

\newtheorem{theorem}{Theorem}
\newtheorem{lemma}{Lemma}
\newtheorem{remark}{Remark}

\newcommand{\reals}{\mathbb{R}}
\newcommand{\complexes}{\mathbb{C}}
\newcommand{\conj}{^*}
\newcommand{\abs}[1]{\left|#1\right|}
\newcommand{\CN}[2]{\mathcal{CN}\left( #1,#2 \right)}
\newcommand{\vect}[1]{\textbf{#1}}
\newcommand{\EX}[2][]{\mathbb{E}_{#1}\left[#2\right]}

\begin{document}

\title{Over-the-Air Computation with Power Amplifier Nonlinearity
}

\author{
\IEEEauthorblockN{
David Nordlund\IEEEauthorrefmark{1},
Ayça Özçelikkale\IEEEauthorrefmark{2},
and Zheng Chen\IEEEauthorrefmark{1}
}
\IEEEauthorblockA{\IEEEauthorrefmark{1}
Department of Electrical Engineering, Link\"oping University, Sweden\\
\{david.nordlund, zheng.chen\}@liu.se
}
\IEEEauthorblockA{\IEEEauthorrefmark{2}
Department of Electrical Engineering, Uppsala University, Sweden\\
ayca.ozcelikkale@angstrom.uu.se
}
}

\maketitle

\begin{abstract}
Over-the-air computation (AirComp) enables efficient wireless data aggregation by exploiting the superposition of signals transmitted by multiple source nodes, but its performance is sensitive to hardware impairments such as power amplifier (PA) nonlinearity. 
In this paper, we develop an analytical framework for coherent AirComp with PA nonlinearities and derive a closed-form expression for the receiver mean-squared error (MSE) under a memoryless polynomial PA model. Our numerical results show that power control optimized under the conventional linear-PA assumption can exhibit rapidly increasing MSE as the transmit power budget increases, whereas incorporating PA nonlinearity into the power control design prevents this degradation and maintains low MSE over the considered power range.
\end{abstract}

\begin{IEEEkeywords}
Over-the-air computation, power control, hardware impairments, nonlinearity
\end{IEEEkeywords}

\section{Introduction}
Driven by the demand for efficient data aggregation in massive Internet-of-Things (IoT) networks, over-the-air computation (AirComp), a joint communication-and-computation framework, has attracted significant attention \cite{otaZhu,ota-old-new,10092857, zhang2024fluid}.
By exploiting the superposition property of wireless channels, AirComp enables the direct computation of nomographic functions in the air \cite{nazer2007computation,otaGoldenbaum}. This allows data aggregation to be performed simultaneously with signal transmission, enabling efficient resource utilization that is largely independent of the number of source nodes.

Most existing studies on AirComp use coherent analog processing with linear scaling, which requires the source nodes to know the channel state information (CSI) for amplitude scaling and phase alignment. Different power control strategies have been explored, some ensuring unbiased estimation and others relaxing this constraint for better mean-squared error (MSE) performance.
For static channels with per-user average power constraints, the MSE-optimal strategy is known to take a partial channel inversion form: source nodes with sufficiently strong channels invert the channel exactly, while the remaining nodes transmit at full power \cite{caoOptimized}.
Related MSE-optimal power allocation has also been studied under alternative assumptions, including sum-power constraints \cite{zang2020over} and correlated source data \cite{wang2024optimal}.

However, practical AirComp systems are also affected by impairments such as phase noise \cite{dahl2024over} and imperfect CSI \cite{evgenidis2023over,chen2022over}, which can disrupt phase alignment and degrade channel inversion.
Beyond these effects, existing processing schemes typically assume perfectly linear signal amplification, which may not hold in practical transmitters. In particular, power amplifiers (PAs) exhibit nonlinear behavior when operating close to their power limits \cite{morgan2006generalized}, potentially causing significant estimation errors in AirComp systems where some nodes operate at or near saturation.
PA nonlinearity is commonly mitigated at the hardware level through digital predistortion (DPD), which linearizes the transmit chain prior to amplification \cite{morgan2006generalized}.
An alternative strategy that requires no additional transmitter hardware is power control at the source nodes, which is particularly attractive for resource- and battery-constrained IoT devices. 
In the context of over-the-air federated learning, peak-power constraints imposed by PA nonlinearity have been addressed via amplitude clipping and frequency-domain filtering \cite{bielefeld2025signal}.

Existing studies have thus addressed PA-related peak-power and clipping effects, but the impact of PA nonlinearity on AirComp power control design remains largely unexplored.
In this work, we introduce a general framework for coherent AirComp with PA nonlinearity, and derive a closed-form expression for the MSE at the receiver under a memoryless polynomial PA model.
Using this expression, our numerical results show that power control designed for linear PAs, although optimal under the ideal linear model, results in rapidly increasing MSE as the transmit power budget increases when PA nonlinearity is present. We further show that accounting for PA nonlinearity in the power control design, even without explicitly modeling hardware saturation, limits the MSE growth and substantially outperforms the linear-PA baseline at high power.
\section{Coherent Over-the-Air Computation with Power Amplifier Nonlinearity}
\label{sec:system_model}

We consider a system with $K$ source nodes indexed by $k \in \mathcal{K}= \{1,2, \ldots ,K\}$, transmitting data to a central server over a multiple-access channel with channel coefficients $(h_1, \ldots, h_K) = \vect{h} \in \complexes^K$. 
The source data are assumed to be normalized with distribution $(s_1, \ldots, s_K) = \vect{s} \sim \CN{\vect{0}}{\mathbf{I}_K}$.
The objective is for the server to estimate the sum of the distributed source data, $f(\vect{s})= \sum_{k\in \mathcal{K}} s_k \in \complexes$, via AirComp.
As is commonly done in the literature (see, e.g., \cite{du2023gradient, wang2024optimal, broadband ,caoOptimized}), CSI is assumed to be known, enabling coherent pre-processing at the source nodes. 
Specifically, each node pre-compensates for the channel phase shift so that the transmitted signals combine coherently at the receiver.

Accordingly, node $k$ wants to transmit the precoded symbol
\begin{equation}
    x_k = \phi(s_k)= \frac{h_k^*}{|h_k|} \sqrt{\rho_k}  s_k,
    \label{encoding}
\end{equation}
where $(\cdot)^*$ denotes the complex conjugate and $\rho_k \in \reals$ is the transmit power of node $k$, subject to average power constraints 
\begin{equation}
    \label{eq:power_constraint}
    \mathbb{E}[\abs{x_k}^2] = \rho_k \leq P
\end{equation}
at each source node. 

In practical systems, PAs introduce nonlinear distortions such that the actually transmitted symbols become $g(x_k)\in \complexes$. 
The server receives
\begin{equation}
    y= \sum_{k \in \mathcal{K}} h_k g(x_k) + w, 
\end{equation}
where $w\sim \CN{0}{\sigma^2}$ is additive noise. Based on $y$, the server constructs an estimate $\hat f$ of the target function.

Conventional AirComp processing schemes are typically derived under the assumption of perfectly linear PAs, i.e.,
\begin{equation}
  g(x_k)= g_\mathrm{L}(x_k) = x_k.  
\end{equation}
In this work, we show that such schemes can suffer substantial performance degradation in the presence of PA nonlinearity. To this end, we first review several common AirComp processing strategies under the ideal linear PA assumption. We then incorporate nonlinear PA distortions into these schemes and characterize the estimation performance with the MSE.

\subsection{Linear Power Amplifiers}
\label{sec:system_model_linear}
Under the linear PA assumption, the server receives 
\begin{equation}
    y= \sum_{k \in \mathcal{K}} h_k x_k +w = \sum_{k \in \mathcal{K}} \abs{h_k} \sqrt[]{\rho_k} s_k + w.
\end{equation}
After post-processing with $\psi(y) = \gamma y$, the estimated computation result is
\begin{equation}
\hat{f}=\gamma y= \gamma \sum_{k \in \mathcal{K}} \abs{h_k} \sqrt{\rho_k} s_k + \tilde{w},
\end{equation} 
where $\gamma \in \reals$ is a linear scaling factor and  $\tilde{w} \sim \CN{0}{\gamma^2 \sigma^2}$. 
The system performance is measured by the MSE, given by
\begin{align}
  e (\rho_1, \hdots, \rho_K)= \EX[w, \vect{s}]{\left\vert f-\hat f \right\vert^2}.
  \end{align}  

In the following, we outline several widely used choices of pre- and post-processing functions (i.e., transmit power control and receive scaling) in the AirComp literature. 

\subsubsection{Full Channel Inversion} 
\label{sec:proc_unbiased}
In most works, the source nodes employ channel inversion to achieve alignment of both the phase and amplitude of the transmitted signals \cite{nazer2007computation}. 
Using the encoding framework in \eqref{encoding}, this introduces the condition
\begin{equation}
\label{eq:unbiased_linear_condition}
	\gamma \abs{h_k} \sqrt{\rho_k} =1, \forall k \in \mathcal{K}.
\end{equation}
One way to satisfy this condition is to select the transmit power control and receive scaling factors as
\begin{align}
	\rho_k = \frac{\eta^2}{\abs{h_k}^2}, \hspace{0.3cm}
	\gamma = \frac{1}{\eta},
\end{align}
for some $\eta \in \reals$ tuned to satisfy the power constraints, i.e., 
\begin{equation}
\label{eq:eta_full_inversion}
	\eta \leq \sqrt{P}\min_{k \in \mathcal{K}}|h_k|.
\end{equation}
In this way, the pre-processing function is $\phi(s_k)=\eta s_k/h_k$, which inverts the channel under the power budget.
The target function estimate becomes
\begin{equation}\label{eq:fhat:fullchannelinversion}
	\hat{f}= f + \frac{1}{\eta} w,
\end{equation}
which is the desired sum plus an effective (zero-mean) noise term that is inversely proportional to the power scaling factor $\eta$. Note that this power control scheme achieves unbiased estimation since $\EX{\hat{f}}=f$.

\subsubsection{Partial Channel Inversion}
\label{sec:proc_partial}
With full channel inversion, 
source nodes with good channel conditions limit their transmit power so that their received signal amplitude matches that of the node with the worst channel. As the number of nodes in the system increases, $\eta$ decreases because the probability that at least one node experiences a particularly poor channel grows, which in turn reduces the effective signal-to-noise ratio (SNR). Therefore, it may be advantageous to allow a small bias in the function estimate to reduce the effective noise variance. In this spirit, several partial channel inversion approaches have been proposed in the literature.

Let $\mathcal{K}_h = \{ k :  |h_k|>h_\mathrm{th}\}$ denote the subset of source nodes with sufficiently good channel conditions, where $h_\mathrm{th}$ is a minimum channel quality threshold. 
We refer to $|h_k|>h_\mathrm{th}$ as the channel inversion condition, as the nodes satisfying this condition perform channel-inversion encoding (as outlined in Section~\ref{sec:proc_unbiased}) with power control coefficients:
\begin{equation}
\label{eq:partial_inverting}
	\rho_k = \frac{\eta^2}{\abs{h_k}^2} , \; \forall k \in \mathcal{K}_h,
\end{equation}
where $\eta \leq \sqrt{P}\min_{k\in \mathcal{K}_h}  \abs{h_k}$. 

The remaining source nodes transmit with fixed power using
\begin{equation}
\label{eq:partial_not_inverting}
	\rho_k = \delta P, \; \forall k \in \mathcal{K} \setminus \mathcal{K}_h,
\end{equation}
where $\delta \in [0,1]$ denotes the fraction of the available transmit power used for actual transmission. The receive scaling factor remains unchanged, i.e., 
\begin{equation}
\label{eq:partial_gamma}
    \gamma = \frac{1}{\eta}.
\end{equation}
\begin{remark}
\label{remark:general_processing_scheme}
This formulation unifies and extends existing variants of (partial) channel inversion schemes. In particular, $h_\mathrm{th} = 0$ corresponds to full channel inversion, while $\delta = 0$ reduces to the ``truncated channel inversion'' scheme, where source nodes with poor channel quality do not participate in the aggregation (e.g., \cite{broadband}). At the other extreme, $\delta = 1$ recovers the ``threshold-based power control'' scheme, in which such source nodes transmit with maximum power (e.g., \cite{caoOptimized}). The intermediate regime $0<\delta<1$ has, to the best of our knowledge, not been studied in the literature. As we show later, when accounting for PA nonlinearity, $\delta$ becomes an important design parameter, and its optimal value may lie strictly within this range.
\end{remark}

\subsection{Nonlinear Power Amplifiers}
\label{sec:system_model_nonlinear}
We now extend the coherent AirComp framework to account for nonlinear PA distortions.
A widely used model for nonlinear PAs is the memoryless polynomial model \cite{faulkner2002spectral,demir2026network, handel2020impact}
\begin{equation}
\label{eq:PA_model_general}
	g(x_k) = g_\mathrm{P}(x_k) = \sum_{l=0}^L \alpha_{2l+1} x_k \abs{x_k}^{2l},
\end{equation}
where $2L+1$ is the model order, and $\{ \alpha_{2l+1}\}_{l=0}^L \subset \complexes$ are the polynomial coefficients.

Substituting \eqref{encoding} into \eqref{eq:PA_model_general}, the transmitted signal at the $k$-th node can be expressed as
\begin{align}
\label{eq:gp}
    g_\mathrm{P}(x_k)
    &= \sum_{l=0}^{L} \alpha_{2l+1} \rho_k^{l+\frac12} \frac{h_k^*}{|h_k|} s_k \abs{s_k}^{2l}.
\end{align}
The server receives
\begin{align}
    y
    &= \sum_{k \in \mathcal{K}} \sum_{l=0}^{L} \alpha_{2l+1} \rho_k^{l+\frac12} \abs{h_k} s_k \abs{s_k}^{2l} + w,
\end{align}
and finally computes
\begin{align}
    \hat{f} 
    &= \sum_{k \in \mathcal{K}} \sum_{l=0}^{L} c_{k,l} s_k \abs{s_k}^{2l} + \tilde{w},
\end{align}
where
\begin{equation}
    \label{eq:MSE_c_coeff}
    c_{k,l} = \alpha_{2l+1} \gamma \abs{h_k} \rho_k^{l+\frac{1}{2}} \in \complexes.
\end{equation}

For a given PA model $X$, let $e_X(\cdot)$ denote the corresponding MSE, i.e., the MSE obtained when $g(\cdot)=g_X(\cdot)$.
The following holds for $e_\mathrm{P}(\cdot)$ and a given realization of the channel coefficients $\vect{h}$.

\begin{theorem}
\label{thm:MSE}
    Let $c_{k,l}$ be given by \eqref{eq:MSE_c_coeff}, and define
    \begin{align}
        d_{k,0} &= 1-c_{k,0}, \\
        d_{k,l} &= -c_{k,l}, \quad l \geq 1.
    \end{align}
    For source data distributed as $\vect{s} \sim \CN{\vect{0}}{\mathbf{I}_K}$, the MSE is
    \begin{align}
       e_\mathrm{P}(\rho_1, \hdots, \rho_K) = \gamma^2 \sigma^2 + \sum_{k \in \mathcal{K}} \sum_{l=0}^{L} \sum_{l^\prime=0}^{L} (l+l^\prime+1)! \, d_{k,l} \conj d_{k,l^\prime} \, . \label{eq:analytical_MSE}
    \end{align}  
\end{theorem}
\begin{proof}
   See Appendix~\ref{app:proof_th1}.
\end{proof}
Theorem~\ref{thm:MSE} holds for coherent AirComp under an arbitrary power allocation $\rho_k$ at the source nodes, i.e., it does not presume a specific power control scheme. We now specialize this result to the partial channel inversion scheme of Section~\ref{sec:proc_partial}, parametrized by $(\delta,\eta,h_\mathrm{th})$, to characterize how the MSE depends on these design parameters.

Recall that under this scheme, $\rho_k$ takes one of two forms depending on whether $k \in \mathcal{K}_h$, where $\mathcal{K}_h$ is determined by $h_\mathrm{th}$. Substituting \eqref{eq:partial_inverting}, \eqref{eq:partial_not_inverting}, and \eqref{eq:partial_gamma} into \eqref{eq:MSE_c_coeff} gives
\begin{equation}
\label{eq:c_coeff_partial}
c_{k,l} =
\begin{cases}
\alpha_{2l+1} \dfrac{\eta^{2l}}{|h_k|^{2l}}, & k \in \mathcal{K}_h, \\
\dfrac{\alpha_{2l+1} |h_k| (\delta P)^{l+\frac{1}{2}}}{\eta}, & k \in \mathcal{K}\setminus\mathcal{K}_h.
\end{cases}
\end{equation}
Substituting \eqref{eq:c_coeff_partial} into \eqref{eq:analytical_MSE} thus expresses the MSE as a function of $(\delta,\eta,h_\mathrm{th})$. We denote this MSE by $e_\mathrm{P}(\delta,\eta,h_\mathrm{th})$.

The choice of power control parameters $\eta$, $\delta$, and $h_\mathrm{th}$ jointly affects system performance by influencing both the effective noise and the nonlinear distortion.

We are interested in the power control parameters that minimize the MSE, i.e., in solving
\begin{align}
\label{eq:opt_nonlinear}
\min_{\delta \in [0,1], \eta \geq 0, h_\mathrm{th}} \quad & e_\mathrm{P}(\delta,\eta,h_\mathrm{th}) \\
\text{s.t.} \quad & \eta \leq \sqrt{P}\min_{k \in \mathcal{K}_h}{|h_k|}. \nonumber
\end{align}
This problem is, in general, nonconvex.
Our objective is to characterize the performance of PA-aware power control rather than to develop a globally optimal solution to this nonconvex problem. We therefore resort to numerical optimization.
Although $h_\mathrm{th}$ is a continuous variable, for fixed $(\eta,\delta)$, the MSE is piecewise constant with respect to $h_\mathrm{th}$ and changes only when $h_\mathrm{th}$ crosses the channel magnitude of a source node. Thus, only $K+1$ intervals need to be considered. For each candidate interval, the remaining optimization is two-dimensional in $(\eta,\delta)$ and bounded, making numerical optimization feasible.

Solving \eqref{eq:opt_nonlinear} requires knowledge of the polynomial coefficients. In practice, these coefficients can be estimated using well-established techniques from the PA modeling literature, see, e.g., \cite{morgan2006generalized}.

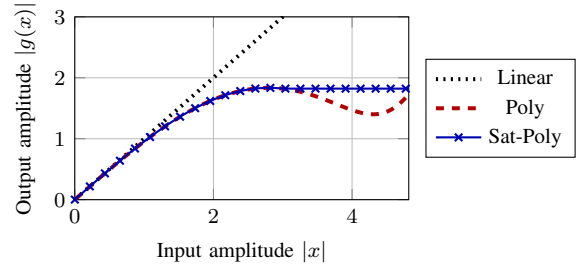
\begin{figure}%
    \centering
    \begin{tikzpicture}
        \begin{axis}[
            width  = 6cm,
            height = 4cm,
            xmin = 0,  xmax = 4.8175,
            ymin = 0,  ymax = 3,
            xlabel = {Input amplitude $|x|$},
            ylabel = {Output amplitude $|g(x)|$},
            tick style={major tick length=4pt, minor tick length=0pt},
            grid=major,
            legend style={
                at={(1.05,0.5)},
                anchor=west,
                cells={align=left},
                legend columns=1,
            },
        ]
        
        \addplot[
            thick,
            dotted,
            color = black,
            line width = 1.5pt,
        ] table[x=x, y=yl, col sep=space] {data/general/PA.dat};
        \addlegendentry{Linear}

        \addplot[
            thick,
            dashed,
            color = red!70!black,
            line width = 1.5pt,
        ] table[x=x, y=ynl, col sep=space] {data/general/PA.dat};
        \addlegendentry{Poly}

        \addplot[
            thick,
            solid,
            color = blue!70!black,
            mark = x,
            mark repeat = 15,
        ] table[x=x, y=ynls, col sep=space] {data/general/PA.dat};
        \addlegendentry{Sat-Poly}
        
        \end{axis}
    \end{tikzpicture}
    \caption{Amplitude input-output characteristics of the PA for the linear, third-degree polynomial (Poly), and saturated third-degree polynomial (Sat-Poly) models. 
   }
    \label{fig:PA_behaviour}
\end{figure}

\section{Simulation Results}
The system consists of $K=10$ source nodes uniformly and randomly distributed in an annulus with radii  $1\,$m and $100\,$m, with the server at the center. The channel is subject to Rayleigh fading, so that all $h_k \sim \CN{0}{d_k^{-4}}$, where $d_k$ is the link distance of node $k$. 
We use a third-order polynomial model ($L=1$) with $\alpha_1=1$ and  $\alpha_3 \in \{0, \, 0.0479\,\mathrm{e}^{-j\,2.816}\}$, where $\alpha_3 =0$ corresponds to the linear case and $\alpha_3=0.0479\,\mathrm{e}^{-j\,2.816}$ is an example of the nonlinear case, see \cite{faulkner2002spectral}.
Polynomial models typically provide an accurate representation of PA behavior over the input-amplitude range used for parameter fitting, see, e.g., \cite{faulkner2002spectral}. However, at large input amplitudes, where practical PAs exhibit saturation, the model may no longer accurately characterize the PA behavior, as illustrated in Fig.~\ref{fig:PA_behaviour}.
To better reflect practical PA behavior, we additionally incorporate a saturation region in our simulations. Motivated by the PA response shown in Fig.~\ref{fig:PA_behaviour}, we enforce hard clipping for input amplitudes $|x_k| > x_\mathrm{sat} = 3$, and evaluate the resulting system performance under this more realistic Sat-Poly model. That is, the system is simulated with 
\begin{align} 
\label{eq:PA_PS}
g(x_k)= g_\mathrm{SP}(x_k) \triangleq
    \begin{cases}
    g_\mathrm{P}(x_k), & \text{if } \abs{x_k} \leq x_\mathrm{sat}, \\
    g_\mathrm{P}\!\left(x_\mathrm{sat}\dfrac{x_k}{\abs{x_k}}\right),  & \text{if } \abs{x_k} > x_\mathrm{sat},
    \end{cases}
\end{align}
and we report $e_\mathrm{SP}(\delta,\eta,h_\mathrm{th})$ whenever nonlinearity is considered.
All results are averaged over $10^2$ independent realizations of the network and channel, and $10^2$ realizations of the source data and noise.

\subsection{Effect of Power Amplifier Nonlinearity}
\label{sec:res_impact}

In this subsection, we evaluate the impact of PA nonlinearity on the system performance under the processing schemes introduced in Sections~\ref{sec:proc_unbiased} and~\ref{sec:proc_partial}.

In Fig.~\ref{fig:res1}, we plot $e_\mathrm{L}(\delta,\eta,h_{\mathrm{th},0})$ and $e_\mathrm{SP}(\delta,\eta,h_{\mathrm{th},0})$ for the linear and nonlinear PA models, respectively. 
Hence, we quantify the gap between the system performance that one may naively expect under ideal linear PA assumptions and the realistic performance that is obtained under non-ideal PA. The parameter $\eta$ for the different power control schemes is set to its respective upper bound, $\eta = \sqrt{P}\min_{k\in \mathcal{K}_h}  \abs{h_k}$, and the partial channel inversion scheme is evaluated for $\delta=0$ and $\delta=1$. The system is simulated for all combinations of $\alpha_3 \in \{0, 0.0479\,\mathrm{e}^{-j\,2.816}\}$ and $\sigma^2 \in \{0,10^{-8}\}$, corresponding to linear and nonlinear PAs, as well as noiseless and noisy communication settings, respectively. For the partial channel inversion schemes, the channel quality threshold is set to $h_{\mathrm{th},0}=10^{-4}$.

Figs.~\ref{fig:res1_nonlinear_noiseless} and \ref{fig:res1_linear_noiseless} show the noiseless case. By isolating the effect of PA nonlinearity, Fig.~\ref{fig:res1_nonlinear_noiseless} shows that the MSE increases with $P$, as the nonlinear distortion becomes more pronounced at higher transmit powers.
Figs.~\ref{fig:res1_nonlinear_noisy} and \ref{fig:res1_linear_noisy} present the noisy case. In the presence of nonlinearity (Fig.~\ref{fig:res1_nonlinear_noisy}), the MSE exhibits a minimum at a finite value of $P$. 
For small $P$, the noise term $\gamma^2 \sigma^2$ in \eqref{eq:analytical_MSE} dominates, causing the MSE to decrease as $P$ increases. Beyond this minimum, however, nonlinear distortion becomes dominant, and further increases in $P$ degrade system performance.

In both cases, the relative performance of the different power control schemes changes when PA nonlinearity is present. Although full channel inversion achieves the best performance at high input power levels in these cases, the MSE can no longer be reduced arbitrarily when nonlinearity is present. Furthermore, the preferable power scaling strategy for source nodes with insufficient channel quality (either remaining silent or transmitting at maximum power) depends on the input power limit. 
For large $P$, transmitting at the maximum power introduces sufficient distortion to outweigh the performance gains from effective noise reduction. This motivates the intermediate regime $0<\delta<1$ considered in Remark~\ref{remark:general_processing_scheme}, which is explored in the following subsection.

\begin{figure}[!t]
    \centering
    \begin{subfigure}[b]{0.49\linewidth}
        \centering
        \begin{tikzpicture}
            \begin{axis}[
                xlabel={$P$},
                ylabel={MSE},
                tick style={major tick length=4pt, minor tick length=0pt},
                legend style={
                    at={(1.5,1.05)},
                    anchor=south,
                    cells={align=left},
                    legend columns=3,
                    clip=false,
                },
                grid=major,
                width=\linewidth,
                height=3.8cm,
                xmode=log,
                ymode=log,
                every axis plot/.append style={line width=1.5pt},
            ]
                \addplot[color=black, solid,      thick]
                    table [x=P, y=MSE, col sep=space] {data/res1/new/nonlinear_noisy__full_inversion.dat};
                \addlegendentry{Full inversion}
                \addplot[color=red!70!black,   dashed,     thick]
                    table [x=P, y=MSE, col sep=space] {data/res1/new/nonlinear_noisy__truncated_silent.dat};
                \addlegendentry{$\delta=0$}
                \addplot[color=blue!70!black,  dashdotted, thick]
                    table [x=P, y=MSE, col sep=space] {data/res1/new/nonlinear_noisy__truncated_max.dat};
                \addlegendentry{$\delta=1$}
            \end{axis}
        \end{tikzpicture}
        \caption{$(\alpha_3, \sigma^2) =$  \\ \null \hspace{0.45cm}$(0.0479\,\mathrm{e}^{-j\,2.816},10^{-8})$}
        \label{fig:res1_nonlinear_noisy}
    \end{subfigure}
    \begin{subfigure}[b]{0.49\linewidth}
        \centering
        \begin{tikzpicture}
            \begin{axis}[
                xlabel={$P$},
                ylabel={},
                tick style={major tick length=4pt, minor tick length=0pt},
                grid=major,
                width=\linewidth,
                height=3.8cm,
                xmode=log,
                ymode=log,
                ymin = 1e-1,
                every axis plot/.append style={line width=1.5pt},
            ]
                \addplot[color=black, solid,      thick]
                    table [x=P, y=MSE, col sep=space] {data/res1/new/nonlinear_noiseless__full_inversion.dat};
                \addplot[color=red!70!black,   dashed,     thick]
                    table [x=P, y=MSE, col sep=space] {data/res1/new/nonlinear_noiseless__truncated_silent.dat};
                \addplot[color=blue!70!black,  dashdotted, thick]
                    table [x=P, y=MSE, col sep=space] {data/res1/new/nonlinear_noiseless__truncated_max.dat};
            \end{axis}
        \end{tikzpicture}
        \caption{$(\alpha_3, \sigma^2) =$ \\ \null \hspace{0.45cm}$(0.0479\,\mathrm{e}^{-j\,2.816},0)$}
        \label{fig:res1_nonlinear_noiseless}
    \end{subfigure}

    \vspace{0.4em}

    \begin{subfigure}[b]{0.49\linewidth}
        \centering
        \begin{tikzpicture}
            \begin{axis}[
                xlabel={$P$},
                ylabel={MSE},
                tick style={major tick length=4pt, minor tick length=0pt},
                grid=major,
                width=\linewidth,
                height=3.8cm,
                xmode=log,
                ymode=log,
                every axis plot/.append style={line width=1.5pt},
            ]
                \addplot[color=black, solid,      thick]
                    table [x=P, y=MSE, col sep=space] {data/res1/new/linear_noisy__full_inversion.dat};
                \addplot[color=red!70!black,   dashed,     thick]
                    table [x=P, y=MSE, col sep=space] {data/res1/new/linear_noisy__truncated_silent.dat};
                \addplot[color=blue!70!black,  dashdotted, thick]
                    table [x=P, y=MSE, col sep=space] {data/res1/new/linear_noisy__truncated_max.dat};
            \end{axis}
        \end{tikzpicture}
        \caption{$(\alpha_3, \sigma^2) = (0,10^{-8})$}
        \label{fig:res1_linear_noisy}
    \end{subfigure}
    \begin{subfigure}[b]{0.49\linewidth}
        \centering
        \begin{tikzpicture}
            \begin{axis}[
                xlabel={$P$},
                ylabel={},
                tick style={major tick length=4pt, minor tick length=0pt},
                grid=major,
                width=\linewidth,
                height=3.8cm,
                xmode=log,
                ymode=log,
                ymax = 100,
                every axis plot/.append style={line width=1.5pt},
            ]
                \addplot[color=black, solid,      thick]
                    table [x=P, y=MSE, col sep=space] {data/res1/new/linear_noiseless__full_inversion.dat};
                \addplot[color=red!70!black,   dashed,     thick]
                    table [x=P, y=MSE, col sep=space] {data/res1/new/linear_noiseless__truncated_silent.dat};
                \addplot[color=blue!70!black,  dashdotted, thick]
                    table [x=P, y=MSE, col sep=space] {data/res1/new/linear_noiseless__truncated_max.dat};
            \end{axis}
        \end{tikzpicture}
        \caption{$(\alpha_3, \sigma^2) = (0,0)$}
        \label{fig:res1_linear_noiseless}
    \end{subfigure}

    \caption{MSE of different processing schemes with and without PA nonlinearity ($\alpha_3$), and with and without channel noise ($\sigma^2$).}
    \label{fig:res1}
\end{figure}
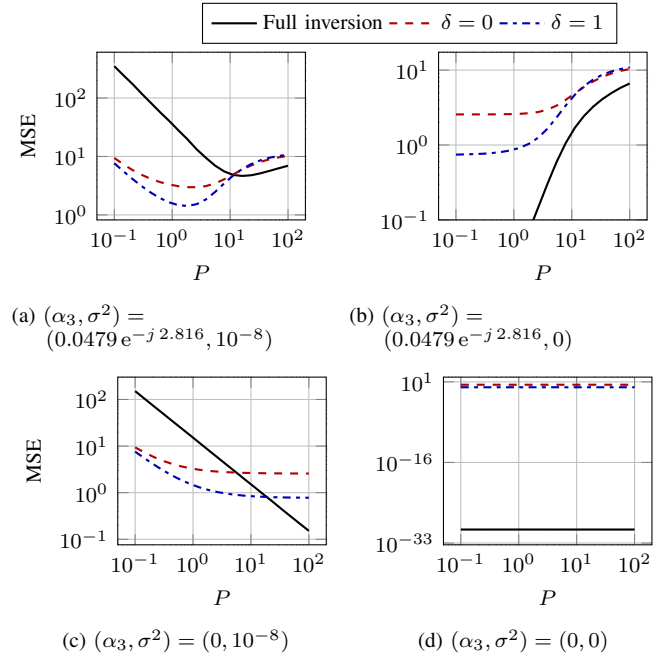

\subsection{Nonlinearity-Aware Power Control}
\label{sec:res:power_control}
Fig.~\ref{fig:res} shows the MSE of the proposed nonlinearity-aware power control scheme, given by the solution to \eqref{eq:opt_nonlinear}, two baseline schemes from the linear AirComp literature, and a genie benchmark. The linear baselines comprise the optimal power control scheme from \cite{caoOptimized} (Lin-Opt), given by the solution to
\begin{align}
\label{eq:opt_linear}
\min_{\delta \in [0,1],\,\eta \geq 0,\,h_{\mathrm{th}}} \quad & e_\mathrm{L}(\delta,\eta,h_{\mathrm{th}}) \\
\text{s.t.} \quad & \eta \leq \sqrt{P}\min_{k \in \mathcal{K}_h}|h_k|,
\nonumber
\end{align}
and truncated channel inversion with $\delta=0$ (Truncated), as in, e.g., \cite{broadband}.
The genie accounts for PA saturation by empirically approximating the solution to 
\begin{align}
\label{eq:opt_genie}
\min_{\delta \in [0,1],\,\eta \geq 0,\,h_{\mathrm{th}}}
\quad &
e_\mathrm{SP}(\delta,\eta,h_\mathrm{th}) \\
\text{s.t.} \quad &
\eta \leq \sqrt{P}\min_{k \in \mathcal{K}_h} |h_k|.
\nonumber
\end{align}
For each channel realization, the MSE under the PA model in \eqref{eq:PA_PS} is estimated over a batch of source realizations for all $K+1$ candidate $h_\mathrm{th}$ intervals and all $(\eta,\delta)$ pairs on a $25\times25$ logarithmically spaced grid. The power control setting yielding the lowest empirical MSE is selected and fixed over the entire batch. The additive noise level is set to $\sigma^2=10^{-8}$.

We emphasize that all curves report the true system performance under nonlinearity, i.e., $e_\mathrm{SP}(\delta,\eta,h_\mathrm{th})$, evaluated at the values of $\delta$, $\eta$, and $h_\mathrm{th}$ obtained from the respective optimization problems above.

For small $P$, Lin-Opt and the nonlinearity-aware scheme achieve essentially identical performance, as the PA operates predominantly in its linear region. As $P$ increases, however, the two schemes increasingly diverge. The MSE of Lin-Opt increases due to the growing nonlinear distortion, whereas the proposed scheme adapts its power control parameters to the nonlinear PA and converges to a finite MSE. This demonstrates that the power control parameters that are optimal under the linear PA model need not remain suitable when nonlinear distortion becomes significant.
Notably, Lin-Opt also becomes inferior to the simpler Truncated scheme at sufficiently large $P$. This illustrates that power control optimized under the assumption of a linear PA can become detrimental when nonlinear distortion dominates. In contrast, the proposed scheme explicitly accounts for the nonlinear PA characteristic when selecting the power control parameters, allowing it to avoid the continued increase in nonlinear distortion observed for the linear-PA baseline.

The proposed scheme exhibits a similar behavior to the genie as $P$ increases, with the MSE decreasing at low $P$ and converging to a finite value at high $P$. Although a performance gap to the genie remains across the considered range, it is relatively small.
At low $P$, where PA saturation is negligible, the gap primarily results from the genie optimizing the power control parameters using the actual source realizations in the batch, whereas the proposed scheme optimizes the MSE based on the statistical source and PA models. At high $P$, the additional knowledge of the actual PA saturation behavior further contributes to the genie's advantage.

Nevertheless, the proposed scheme captures the key behavior of the genie, namely the stabilization of the MSE at high $P$, while requiring only statistical knowledge of the source distribution and PA model.
This is particularly notable because the proposed optimization does not explicitly account for the saturation region used in the simulations, yet still avoids the unbounded performance degradation observed for power control designed under the linear PA assumption.

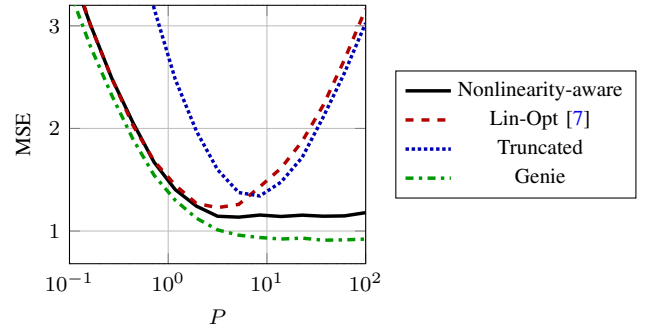
\begin{figure}[!t]
    \centering
    \begin{tikzpicture}
        \begin{axis}[
            xlabel={$P$},
            ylabel={MSE},
            tick style={major tick length=4pt, minor tick length=0pt},
            legend style={
                at={(1.1,0.5)},
                anchor=west,
                cells={align=left},
                legend columns=1,
            },
            grid=major,
            width=5.5cm,
            height=5cm,
            xmode=log,
            xmax=100,
            ymax=3.2,
            xmin = 0.1,
            every axis plot/.append style={line width=1.3pt},
        ]
            \addplot[color=black, solid]
                table [x=P, y=MSE, col sep=space]
                {data/res/analytical.dat};
            \addlegendentry{Nonlinearity-aware}
            
            \addplot[color=red!70!black, dashed]
                table [x=P, y=MSE, col sep=space]
                {data/res/linear_opt.dat};
            \addlegendentry{Lin-Opt \cite{caoOptimized}}

            \addplot[color=blue!70!black, densely dotted]
                table [x=P, y=MSE, col sep=space]
                {data/res/baseline.dat};
            \addlegendentry{Truncated}

            \addplot[color=green!60!black, dashdotted]
                table [x=P, y=MSE, col sep=space]
                {data/res/joint.dat};
            \addlegendentry{Genie}
        \end{axis}
    \end{tikzpicture}
    \caption{MSE of the proposed nonlinearity-aware power control scheme, the optimal power control scheme from \cite{caoOptimized} (Lin-Opt), a truncated channel inversion scheme with $\delta=0$ (Truncated), and a saturation-aware genie. 
    }
    \label{fig:res}
\end{figure}

\section{Conclusion}

This paper analyzes the impact of PA nonlinearity on coherent AirComp. Assuming circularly symmetric complex Gaussian source data, we derive a closed-form expression for the receiver MSE under a memoryless polynomial PA model. Numerical results show that power control designed for linear PAs, although optimal under the ideal linear model, can result in rapidly increasing MSE as the transmit power budget increases when PA nonlinearity is present.
In contrast, incorporating PA nonlinearity into the power control design mitigates this degradation and substantially improves performance at high transmit power. The proposed scheme also exhibits the same qualitative high-power behavior as a genie scheme that has knowledge of the actual PA saturation and instantaneous source data, despite relying only on statistical knowledge of the source distribution and polynomial PA model.
These results highlight the importance of accounting for PA nonlinearity when designing power control for practical coherent AirComp systems.

\appendices
\section{Proof of Theorem~1}
\label{app:proof_th1}

Recall that $\vect{s} \sim \CN{\vect{0}}{\mathbf{I}_K}$ and $\tilde{w} \sim \CN{0}{\gamma^2 \sigma^2}$ are independent. Let $s_k = r_k e^{j\theta_k}$ and define $v_{k,l} = s_k |s_k|^{2l}$. We first state the following lemma.

\begin{lemma}
\label{lemma:statistics_general}
$\EX{\tilde{w}v_{k,l}} = 0$, $\EX{v_{k,l}\conj v_{k^\prime,l^\prime}} = 0 \quad \forall k \neq k^\prime$, and $\EX{v_{k,l}\conj v_{k,l^\prime}} = (l+l^\prime+1)!.$
\end{lemma}
\begin{proof}
    Since $s_k \sim \CN{0}{1}$ is circularly symmetric, we may write $s_k=r_ke^{j\theta_k}$, where $\theta_k\sim\mathrm{Uniform}(0,2\pi)$ is independent of $r_k$. Hence, $v_{k,l}=r_k^{2l+1}e^{j\theta_k}$ and $\mathbb{E}[e^{j\theta_k}]=0$ implies $\mathbb{E}[s_k]=\mathbb{E}[v_{k,l}]=0$. Further, $\mathbb{E}[\tilde{w}v_{k,l}]=0$ since $\tilde{w}$ is zero mean and independent of $v_{k,l}$, and for $k\neq k^\prime$,
    $\mathbb{E}[v_{k,l}\conj v_{k^\prime,l^\prime}]
    =\mathbb{E}[v_{k,l}]\mathbb{E}[\conj v_{k^\prime,l^\prime}]=0$
    by independence. Finally,
    $v_{k,l}\conj v_{k,l^\prime}=|s_k|^{2(l+l^\prime+1)}=r_k^{2(l+l^\prime+1)}$, and since $r_k^2\sim\mathrm{Exp}(1)$, whose $n$-th moment is $n! \,$, it follows that
    $\mathbb{E}[v_{k,l}\conj v_{k,l^\prime}]=(l+l^\prime+1)!$.
\end{proof}
Note that $f - \hat{f} = \sum_{k \in \mathcal{K}} \sum_{l=0}^{L} \left(d_{k,l} v_{k,l} \right) - \tilde{w}$, where $d_{k,0}=1-c_{k,0}$ and $d_{k,l}=-c_{k,l}, \, \, l \ge 1$.
Using Lemma~\ref{lemma:statistics_general}, the MSE can be computed as
\begin{align}
    \EX[w, \vect{s}]{\left\vert f-\hat f \right\vert^2}
    =
    \EX[w]{\abs{\tilde{w}}^2} + \EX[\vect{s}]{\abs{\sum_{k \in \mathcal{K}} \sum_{l=0}^{L} d_{k,l} v_{k,l}}^2},
\end{align}
where
\begin{align}
    \EX[\vect{s}]{\abs{\sum_{k \in \mathcal{K}}\sum_{l=0}^{L} d_{k,l} v_{k,l}}^2}
    &= \sum_{k,k^\prime \in \mathcal{K}}\sum_{l=0}^{L} \sum_{l^\prime=0}^{L} d_{k,l}\conj d_{k^\prime,l^\prime} \EX{v_{k,l}\conj v_{k^\prime,l^\prime}}
    \nonumber\\
    &= \sum_{k \in \mathcal{K}} \sum_{l=0}^{L} \sum_{l^\prime=0}^{L} (l+l^\prime+1)! \,d_{k,l}\conj d_{k,l^\prime}.
\end{align}
Since $\EX[w]{\abs{\tilde{w}}^2} = \gamma^2 \sigma^2$, substituting completes the proof.

\bibliographystyle{IEEEtran}
\bibliography{references}

\end{document}